\documentclass[a4paper, 10pt, conference]{ieeeconf}      % Use this line for a4 paper

\IEEEoverridecommandlockouts                              % This command is only needed if 
\usepackage{graphics} % for pdf, bitmapped graphics files
\usepackage{epsfig} % for postscript graphics files
\usepackage{mathptmx} % assumes new font selection scheme installed
\usepackage{times} % assumes new font selection scheme installed
\usepackage{amsmath} % assumes amsmath package installed
\usepackage{amssymb}  % assumes amsmath package installed
\usepackage{subfig}
\usepackage{color}

\graphicspath{{figures/}}

\newtheorem{definition}{Definition}
\newtheorem{exmp}{Example}[section]
\newtheorem{theorem}{Theorem}
\newtheorem{lem}[theorem]{Lemma}
\newtheorem{remark}{Remark}
\newtheorem{corollary}{Corollary}[theorem]

\title{\LARGE \bf
Evolutionary Game Dynamics for Two Interacting Populations
under Environmental Feedback
}

\author{Lulu Gong$^{1}$, Jian Gao$^{2}$, and Ming Cao$^{1}$ % <-this % stops a space
\thanks{The work of Gong and Gao was supported in part by China Scholarship Council (CSC). }% <-this % stops a space
\thanks{$^{1}$ L. Gong and M. Cao are with ENTEG,  Faculty of Science and Engineering,
        University of Groningen, 9747 AG, Groningen, The Netherlands
        {\tt\small l.gong@rug.nl, m.cao@rug.nl}}%
\thanks{$^{2}$ J. Gao is with JBI, Faculty of Science and Engineering, University of Groningen, 9747 AG, Groningen, The Netherlands
	{\tt\small jian.gao@rug.nl}}%
}

\begin{document}

\maketitle
\thispagestyle{empty}
\pagestyle{plain}

%%%%%%%%%%%%%%%%%%%%%%%%%%%%%%%%%%%%%%%%%%%%%%%%%%%%%%%%%%%%%%%%%%%%%%%%%%%%%%%%
\begin{abstract}

We study the evolutionary dynamics of games under environmental feedback using replicator equations for two interacting populations. One key feature is to consider jointly the co-evolution of the dynamic payoff matrices and the state of the environment: the payoff matrix varies with the changing environment and at the same time, the state of the environment is affected indirectly by the changing payoff matrix through the evolving population profiles. For such co-evolutionary dynamics, we investigate whether convergence will take place, and if so, how.  In particular, we identify the scenarios where oscillation offers the best predictions of long-run behavior by using reversible system theory. The obtained results are useful to describe the evolution of multi-community societies in which individuals' payoffs and societal feedback interact.

\end{abstract}

%%%%%%%%%%%%%%%%%%%%%%%%%%%%%%%%%%%%%%%%%%%%%%%%%%%%%%%%%%%%%%%%%%%%%%%%%%%%%%%%
\section{INTRODUCTION}
Evolutionary game theory \cite{c1} is widely used to model population dynamics for social and ecological systems since it offers insightful results under meaningful simplification. There exist various classic models, among which the replicator dynamics play a prominent  role \cite{c2}; in this model, the individuals in a well-mixed population play games with each other according to a mutually known payoff matrix . 
When the collective of individuals  can be divided into several populations according to certain constraints, such as local interactions or genetic relationships, multi-population games take place. Then each population can be taken as a cohesive community whose members play games with players from other populations according to corresponding payoff matrices \cite{c3,c4}. The replicator dynamics for  the multi-population  evolutionary games have been developed in  \cite{c5,c6,c7}.

On the other hand, in classic evolutionary game theory, it is generally assumed that the related payoff matrix is constant and not affected by the evolution itself. That is to say, the players will receive pre-determined  payoffs whenever every player has chosen a candidate strategy for the current game play.  In practical scenarios from ecological systems or human society, however, the incentives or punishments for  the individuals in a game  may dynamically change because of changing environment \cite{c8,c9,c10,c11}. 
In public goods dilemmas where the use of common resource, such as water and pasture, is involving, the contest in real population not only modifies the social composition, but also may have a marked effect on the value of subsequent rewards \cite{c12}. If the shared resource is limited, ``the tragedy of the commons" will be the inevitable fate for all players \cite{c13}. However, in some practical situations, the outcome of such evolutionary public goods games can be much more complicated when the shared resources are  not only  affected by the actions of individuals but also act back on the strategic choices of the populations through influencing the payoffs in the game process \cite{c14,c15,c16}. In fact, some actions may be in favor of enhancing the environment while the others weaken it. On the other hand, the environmental context can also in turn have influence on the individual actions by changing the current payoff matrices. Thus a feedback mechanism arises from the environmental state to the game dynamics. 

So far, most of the existing works have been focusing on the dynamics of population profiles under the given fixed payoff matrices, but environmental feedback on game dynamics has not been taken into account adequately. Recently, the environment factors have received surging interests, and their influence is attracting more and more attention \cite{c17}. The co-evolution of strategies and games  have been studied in \cite{c18}, which shows the path to the collapse of cooperation; in particular, the evolvability of payoffs in a fixed environment is discussed, but the direct relation between payoffs and environment is  not considered. It is studied in \cite{c19} how a single population evolves with a changing environmental resource using replicator equations. 
The corresponding system therein exhibits interesting oscillating behavior when the game is characterized by a payoff matrix modified from the prisoner dilemma. Indeed, the single population replicator equations combined with an environment factor take the form of an integrable system which admits constants of motion. Hence, it can be  proved that the periodic orbits exist using level sets of the corresponding energy function.

When the structure of population is extended from a single population to multiple populations, it is of great interest to study whether convergence will take place in such situation, and if so, how. In addition, it will be more challenging to identify complicated dynamics in such a multi-dimensional system. The main contributions of the paper are as follows. 
The dynamic and environment-dependent payoff matrices are utilized to indicate the feedback of environment to game dynamics. Then 
 we consider  the multi-population game dynamics  and the environment state simultaneously and obtain a model of a closed-loop and coupled system. For two interacting populations case, we derive sufficient conditions for the convergence to boundary points. 
 More importantly, we identify the scenarios where oscillation offers the best predictions of long-run behavior by using reversible system theory.

%To model such a feedback mentioned before, the game dynamics and environment state should be considered simultaneously as a comprehensive system. Moreover, 
%the dynamic payoff matrices have to be utilized to indicate the influence of environment factor. At last, the state of the environment should be modified by the actions of populations as well.

The rest of the paper is organized as follows.
Section II introduces how the game and environment dynamics are modeled with necessary background information.  
Section III provides the main results: First, the conditions for convergence in the closed-loop system is analyzed with an illustrative example; then more complicated dynamics are analyzed thoroughly  when the payoffs are in the form of a modified prisoner's dilemma; and lastly, a qualitative approach is used to prove the resulting periodic orbits by applying reversible system theory. In the last section, we conclude our contribution and discuss the directions for further study.

%%%%%%%%%%%%%%%%%%%%%%%%%%%%%%%%%%%%%%%%%%%%%%%%%%%%%%%%

\section{BACKGROUND AND PROBLEM STATEMENT}

\subsection{Replicator dynamics}

As one of the most well-known models in evolutionary game theory, the \textit{replicator dynamics}  describe the evolution of the frequencies of strategies in a well-mixed population \cite{c2}.
Consider a matrix game with a finite set of $m$ pure strategies, $\{s_{1},\ldots,s_{m}\}$, and the entries of the  payoff matrix $A$ are $a_{ij}$ for $i,j=1,\ldots,m$. Let $p_i$ be the proportion of the individuals who choose $s_i$, and denote $ p = [p_1, \ldots, p_m]^T$. Then the single-population replicator dynamics are determined by setting the  growth rate $ \dot{p}_{i}$  proportional to the difference between the expected utility of $s_i$ and the average utility in the whole population, namely
\begin{equation}
\dot{p}_{i}=p_{i}[U_{i}(p)-\bar{U}(p)],
\end{equation}
where $U_{i}(p)=(Ap)_{i}$ is $s_i$'s and  $\bar{U}(p)=p^{T}Ap$ is the average utility. Note that the replicator dynamics (1) are defined on the simplex $\Delta = \{p\vert \sum_i p_i =1\}$. It has been proved \cite{c2} that  $ \Delta $ is invariant under (1).

Now we extend the single-population replicator dynamics (1) to the multi-population case. Consider an $n$-population system and its replicator dynamics
\begin{equation}\label{eqn:replicator}
\dot{p}_{i}^{k}=p_{i}^{k}[U_{i}^{k}(\mathbf{p})-\bar{U}^{k}(\mathbf{p})], 
\end{equation} 
where   $ p_{i}^{k} $ is the proportion of individuals in population $k$, $ k=1,\ldots,n $, who are currently using $s_i$ and $ \mathbf{p}=[p^{1},\ldots,p^{n}]^T $, where $p^k$ is determined by the proportions of the players of different strategies in population $k$. It can be easily verified that the state space of the whole population, which becomes a  polyhedron, denoted by $ \Theta $, is invariant under (2), so do $ \Theta $'s interior and boundary respectively.

%------------------------------------------------

\subsection{Dynamic payoff matrices}
To model the feedback of the environment to game dynamics, we  introduce the dynamic payoff matrices. According to \cite{c11,c19}, changes in the richness of biological environments, or the economic situation of governing bodies in social settings, sometimes can be captured by varying the payoff matrices by a multiplication factor. Therefore, we consider the  matrix games where the payoff varies with a scalar environment variable $r$, which will be explained in detail in the coming subsection. In particular, we assume that the entries of each payoff matrix depend affinely  on $ r $, i.e., 
\begin{equation}\label{eqn:dynamic payoff}
\begin{bmatrix}
a_{11}r+b_{11}&\ldots&a_{1m}r+b_{1m}\\
\vdots&\ddots&\vdots\\
a_{m1}r+b_{m1}&\ldots&a_{mm}r+b_{mm}
\end{bmatrix}.
\end{equation}

%------------------------------------------------
\subsection{Environmental factor}

To address in depth the environment's influence on game dynamics and vice versa,  the environment in \cite{c19} is modeled by a scalar function coupled with game dynamics. We will use this mechanism to consider the multi-population games and environment dynamics. To be more specific, we represent the change of the environmental resource $ r $ by a continuous scalar function, i.e.,
\begin{equation}\label{eqn:envi}
\dot{r}=r(1-r)h(\mathbf{p}),
\end{equation}
where the $r$ is rescaled to be confined in the range $ [0,1] $; $h(\mathbf{p})$ denotes the feedback of the individuals' actions on the environment. The environmental state changes as a result of the states of the populations and the sign of $ h(\mathbf{p})$ determines whether $ r $ will decrease or increase, corresponding to environmental degradation or enhancement, respectively.  

%$ p $ denotes the proportion of individuals enhancing  the resource and  $ 1-p $ is the degrading proportion; $ \theta>0 $ represents the ratio of the enhancement rate to degradation rate.

%------------------------------------------------

%-----------------------------------------------
\subsection{Mathematical model}
For the sake of simplicity,  
we start by considering the case of two interacting populations.  In this case, each individual from population $ k=1,2 $  only interacts with an individual  from the other population. Such games are usually referred to as \textit{bi-matrix games}. In addition, we assume that each population $ k $ has only two strategies $ \{s_{1},s_{2}\} $. Since there are two available strategies in each population, the population profile can be defined by $ p^{k}=[p_{1}^{k}~~1-p_{1}^{k}]^T $ where $ p_{1}^{k} $ is the frequency of strategy $ s_{1} $  and  $ 1-p_{1}^{k}$  the frequency of strategy $ s_{2} $ in population $ k $. Thus, in every population we only need to focus on the evolution of the proportion of one strategy. Hereafter we denote $ x=p_{1}^{1} $ and $  y=p_{1}^{2} $. We consider the strategy set in the context of the cooperation-defection game as a means to motivate our analysis, i.e., the proportion of the first strategy, cooperation, enriches the resource and the proportion of the second strategy, defection, consumes the resource.  

The interactions between the two populations are described by the following dynamic payoff matrices $ A(r) $ and $ B(r) $:
\begin{equation*}
A(r)=\begin{bmatrix}
a(r)&b(r)\\
c(r)&d(r)
\end{bmatrix} ~~\mathrm{and}~~ B(r)=\begin{bmatrix}
e(r)&f(r)\\
g(r)&h(r)
\end{bmatrix}.
\end{equation*}
Note that generally $A(r)\neq B(r) $, and thus the games are asymmetric. More precisely, if a player in population $ 1 $ plays $ s_{1} $ against a player in population $ 2 $ playing $ s_{2} $, then the first player receives the payoff $ b(r) $ and the second player obtains  the payoff $ g(r) $.
Then the utility functions are given by
\begin{equation}
U_{i}^{1}=\left( A(r)\mathbf{y}\right)_{i},\quad U_{i}^{2}=\left( B(r)\mathbf{x}\right)_{i},
\end{equation}
where $ \mathbf{y}=\begin{bmatrix}
y&1-y
\end{bmatrix}^{T} $ and $ \mathbf{x}=\begin{bmatrix}
x&1-x
\end{bmatrix}^{T} $.

Combine the game dynamics (\ref{eqn:replicator}) and environment equation (\ref{eqn:envi}), then we have the corresponding 3-dimensional co-evolutionary system
\begin{equation}\label{eqn:cosys1}
\begin{cases}
&\dot{x}=x(1-x)f(y,r)\\
&\dot{y}=y(1-y)g(x,r)\\
&\dot{r}=r(1-r)h(x,y),
\end{cases}
\end{equation}
where 
\begin{equation}
\begin{aligned}
f(y,r)&=U^{1}_{1}-U^{2}_{1}\\
&=(a(r)-b(r)-c(r)+d(r))y+b(r)-d(r),
\end{aligned}
\end{equation}
and
\begin{equation}
\begin{aligned}
g(x,r)&=U^{1}_{2}-U^{2}_{2}\\
&=(e(r)-f(r)-g(r)+h(r))x+f(r)-h(r).
\end{aligned}
\end{equation}
In addition, we assume $h(x,y)$ is of the following form:
\begin{equation}
\begin{aligned}
h(x,y)&=\theta_1 x-(1-x) +\theta_2 y -(1-y)\\
&=(1+\theta_{1})x+ (1+\theta_{2})y-2,
\end{aligned}
\end{equation}
in which $\theta_{1}>0$ and $\theta_{2}>0$ represent the ratios of the enhancement rate to degradation rate.

Compared with the conventional replicator dynamics (\ref{eqn:replicator}), this model is more complicated because the payoff matrices are environment dependent and  the strategies and environment dynamics are deeply coupled. One immediately can get some basic properties of this new system. Since all the variables, $ x $, $ y $ and $ r $, just vary in the range $ [0, 1] $, the domain for the whole system (\ref{eqn:cosys1}) is exactly the unit cube $ [0, 1]^{3} $ in the Cartesian coordinates. And through the structure of equations (\ref{eqn:cosys1}) one can easily prove the following lemmas.  

\begin{lem}[Fixed points]
The system (\ref{eqn:cosys1}) has eight obvious fixed points which are exactly on the eight corners of the cubic domain, i.e., $ (0,0,0) $, $ (1,0,0) $, $ (0,1,0) $, $ (0,0,1) $, $ (1,1,0) $, $ (1,0,1) $, $ (0,1,1) $ and $ (1,1,1) $. 
\end{lem}

\begin{lem}
[Invariant sets] For  system  (\ref{eqn:cosys1}),  the following statements hold 
\begin{itemize}
	\item the whole domain $ [0,1]^{3} $ is a positively invariant set;
	\item the 6 faces of this cube  are positively invariant sets.
\end{itemize}  
\end{lem}

Because of the above invariance properties, the system dynamics on the 6 faces  are  reduced to planar cases. Therefore it is trivial to analyze them, and we omit discussing these planar dynamics hereafter  and focus on the 3-dimensional dynamics.

\section{MAIN RESULTS}

\subsection{Convergence}
The coexistence of states in both populations has always been the main subject of bi-matrix game. For this reason, we are going to analyze the convergence and non-convergence of the dynamics of (\ref{eqn:cosys1}). 

\begin{theorem}\label{theorem:convergence1}
	When the dynamic payoff matrices (3) admit weakly dominating strategies, the corresponding  co-evolutionary dynamics (6) will always converge to the boundary for the initial conditions $(x_{0},y_{0},r_{0})$.
\end{theorem}
\begin{proof}
Without loss of generality, let the first strategy be the weakly dominating strategy for the first population. Then  the payoff matrix $A(r)$ takes the following form  
	\[ \begin{bmatrix}
	a(r)&b(r)\\
	c(r)&d(r)
	\end{bmatrix},~~a(r)\geq c(r), b(r)\geq d(r). \]
 As all entries of $A(r)$  are affine functions of $r$, the equality signs of the above inequalities hold only when $r$ is exactly at its maximum or minimum. 
Then the  replicator equation corresponding to $x$ becomes 
	\[ \dot{x}=x(1-x)[(a(r)-c(r))y+(b(r)-d(r))(1-y)]. \]
Immediately one can check that the term  in the square bracket is always positive for any initial $r_{0}\in (0,1)$. As a result,  $ x $ will converge to $ 1 $. When $r_{0}= 0$ or $r_{0}= 1$, the dynamics will be restricted in a plane because of the invariance according to Lemma 2.

Similarly one can prove convergence for  $ y $.
When the states $ x $ and $ y $ reach their maximums or minimums, the sign of $ \dot{r} $ will  depend on the parameters $\theta_{1}$ and $\theta_{2}$.  Given these two parameters are constant,  $ r $ will arrive at some fixed point asymptotically.
\end{proof}

We now use an example to illustrate the  convergence.

%However, the theorem above only indicates the sufficient condition for coevolutionary convergence. We show a convergence example without dominating strategies in Example 2.2.  

\begin{exmp}\label{example}
	The Hawk-Dove game with dynamic payoff matrices.
\end{exmp}

A version of the Hawk-Dove game \cite{c1} takes the following environment-dependent form, 
\[ A(r)=\begin{bmatrix}
0.5r&1\\
0&0.5
\end{bmatrix},~~~~ B(r)=\begin{bmatrix}
0.5&0\\
1&0.5r
\end{bmatrix}.\] 

Obviously the cooperation strategy corresponds to the unique Nash Equilibrium (NE) for the first population and the defection strategy  for the second population respectively, whenever $r \in (0,1)$. But when $r$ reaches its maximum or minimum, these two strategies are not NE anymore.  Therefore, they are only weakly dominating strategies.

The corresponding system now becomes
\begin{equation} \label{eqn:coevolsys2}
\begin{cases}
&\dot{x}=x(1-x)(0.5ry-0.5y+0.5)\\
&\dot{y}=y(1-y)(0.5rx-0.5x-0.5r)\\
&\dot{r}=r(1-r)[(1+\theta_{1})x+ (1+\theta_{2})y-2].
\end{cases}
\end{equation}
From the  computation of the Jacobian, one can check the stability of the trivial  corner equilibria. Obviously, five of the equilibria, $ (0,0,0) $, $ (0,0,1) $, $ (1,1,0) $, $ (0,1,1) $ and $ (1,1,1) $, are unstable because all the Jacobians have at least one positive eigenvalue. It is, however, not clear whether the other equilibria are stable or not, since their Jocobians are dependent of parameters.

However, in the first equation $ (0.5ry-0.5y+0.5)=0.5ry+0.5(1-y)\geq 0 $ and the equality holds only when $ y=1, r=0 $. Thus, $ x $ will converge to $ 1 $ if  $ x_{0} \neq 0 $. In contrast, in the second equation $ (0.5rx-0.5x-0.5r)=0.5x(r-1)-0.5r \leq 0 $ and the equality holds only when $ x=0, r=0 $.  So it always holds that $ y $ will converge to $ 0 $  as expected if $ y_{0}\neq 0 $ initially. As a consequence, there are  two more critical sets, i.e., lines $ \{y=1, r=0\} $ and $ \{x=0, r=0\} $, excluding the eight corner equilibrium points.

For $ r $, consider the open set $ (0,1)^{3} $, which contains no equilibrium point as $ (0.5ry-0.5y+0.5)> 0 $ and $ (0.5rx-0.5x-0.5r)< 0 $.
Hence,  every trajectory will go towards a point on the boundary $ \{x=1, y=0 \}$, when the starting point is inside  $ [0,1]^{3} $. Further, on the line $ \{x=1, y=0\} $, we have
\[ \dot{r}= r(1-r)[(1+\theta_{1})x+ (1+\theta_{2})y-2]=r(1-r)(\theta_{1}-1). \]
Thus, the evolution of  $ r $ depends on the parameter $ \theta_{1} $. To sum up, the trajectory will asymptotically  converge to the fixed point $ x=1, y=0, r=1 $ if $ \theta_{1}>1 $ or to the fixed point $ x=1, y=0, r=0 $ if $ \theta_{1}<1 $. For the case $ \theta_{1}=1 $, the trajectory will stay there after it reaches the line $ \{x=1, y=0\} $.
The different situations are showed in Fig. 1.

In this example the states of both of the two populations cannot co-exist no matter what values  the parameters $ \theta_{1}, \theta_{2} $ take. In other words, the environment factor does not make a big difference to the outcome of game dynamics even when the payoff matrices admit only weakly dominating strategies. 

\begin{figure}[htbp!]
	\centering 
	\subfloat[$ \theta_{1}=0.5 $]{\includegraphics[width=.45\columnwidth]{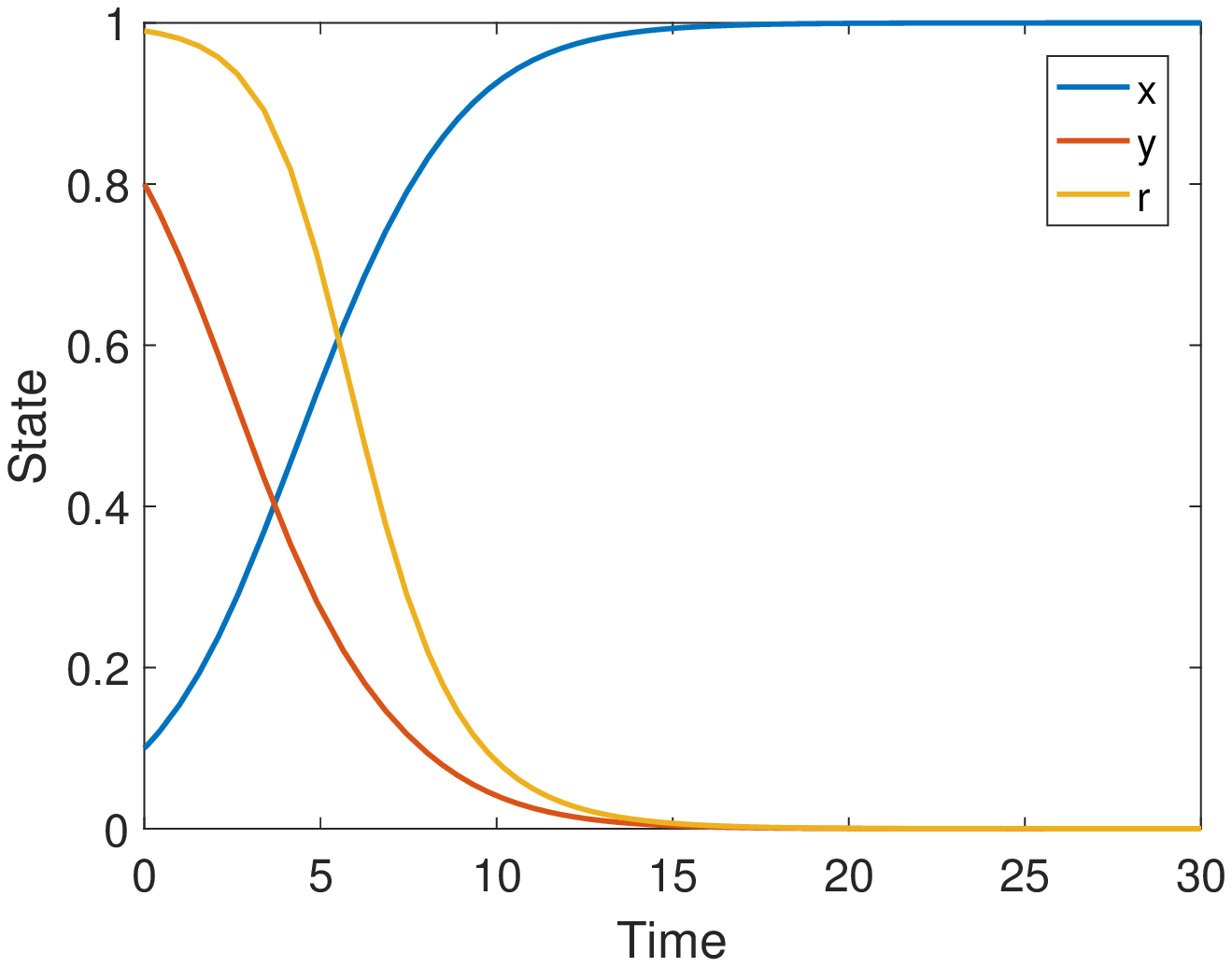}\quad\includegraphics[width=.45\columnwidth]{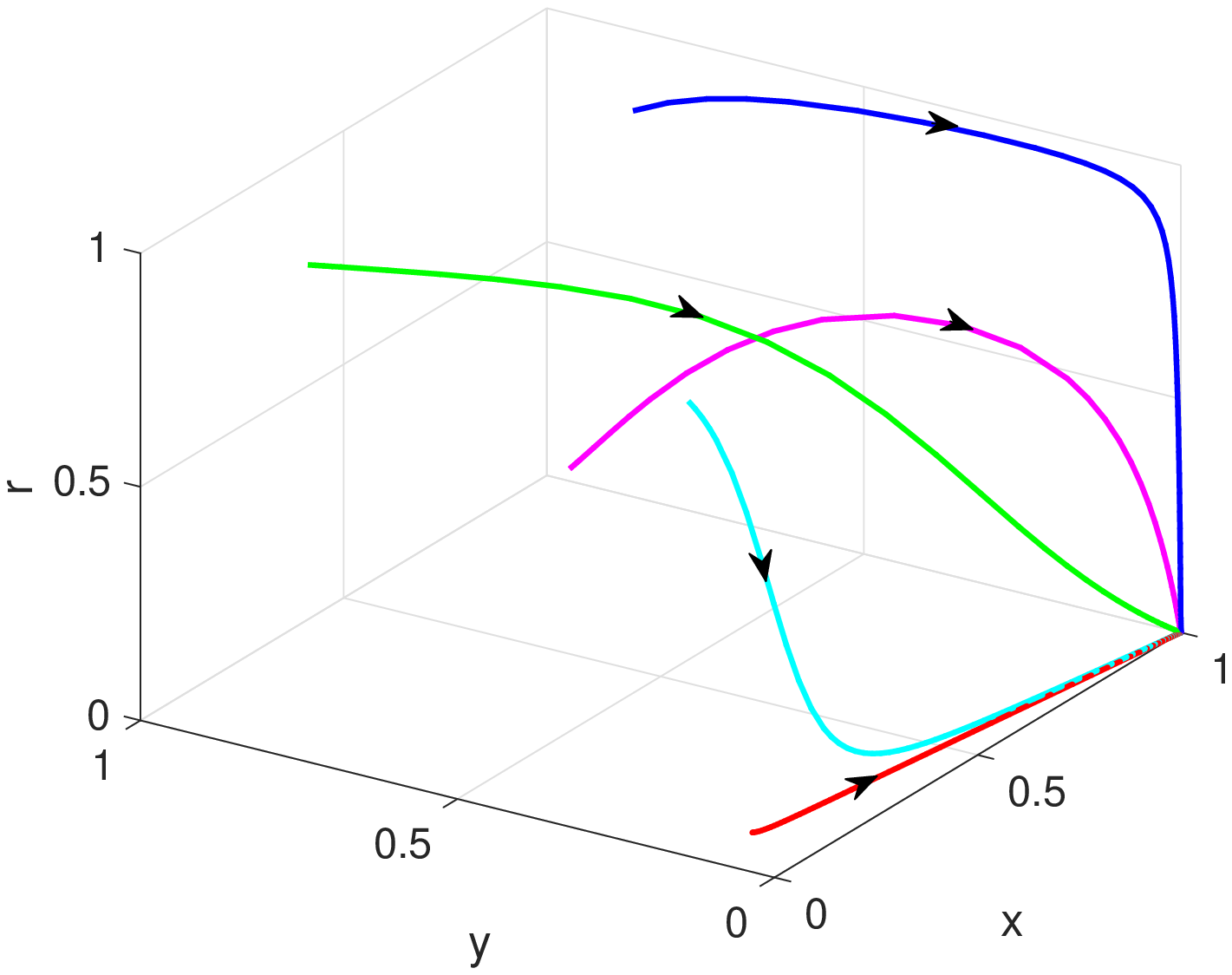}} 
	\\
	\subfloat[$ \theta_{1}=1 $]{\includegraphics[width=.45\columnwidth]{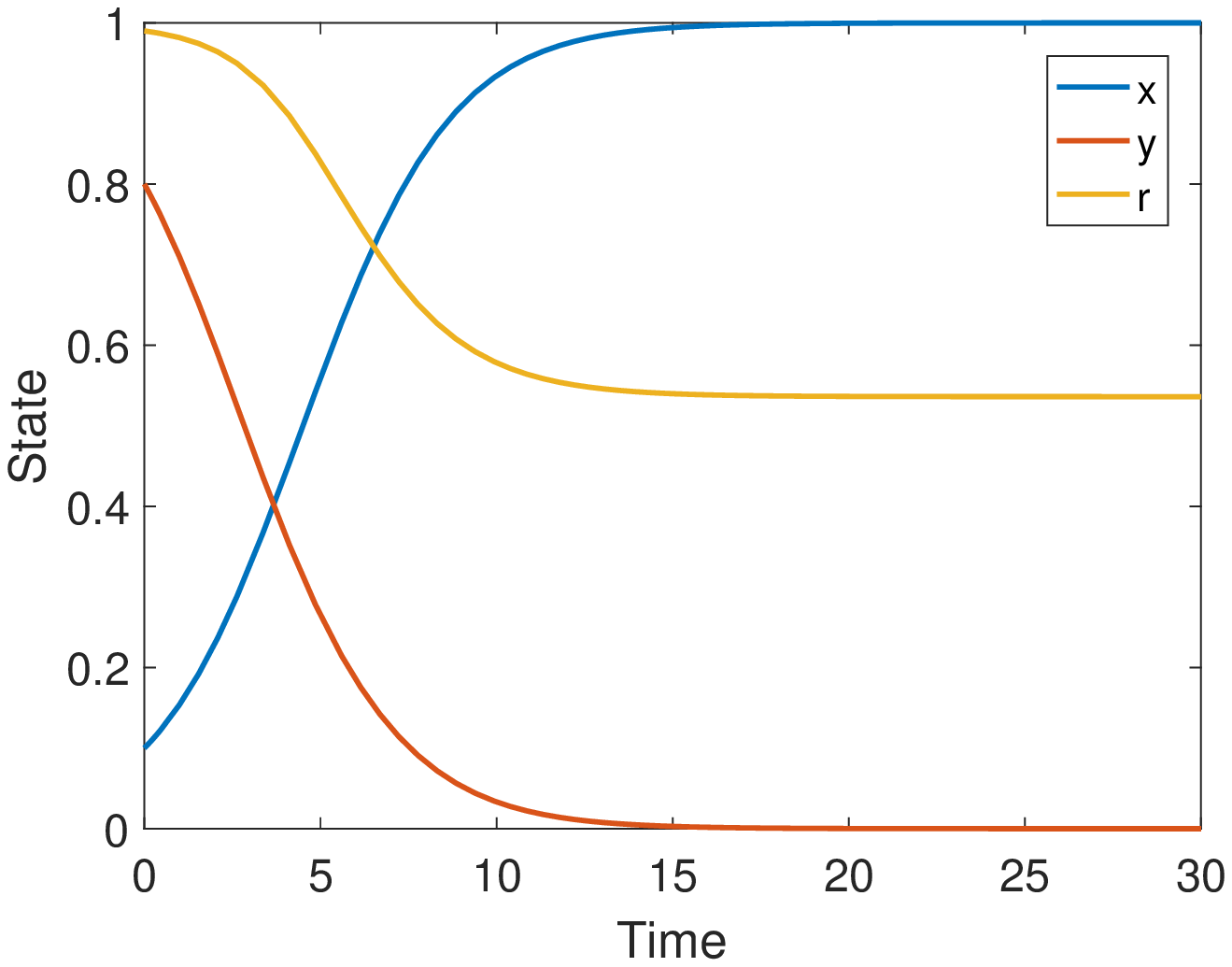}\quad\includegraphics[width=.45\columnwidth]{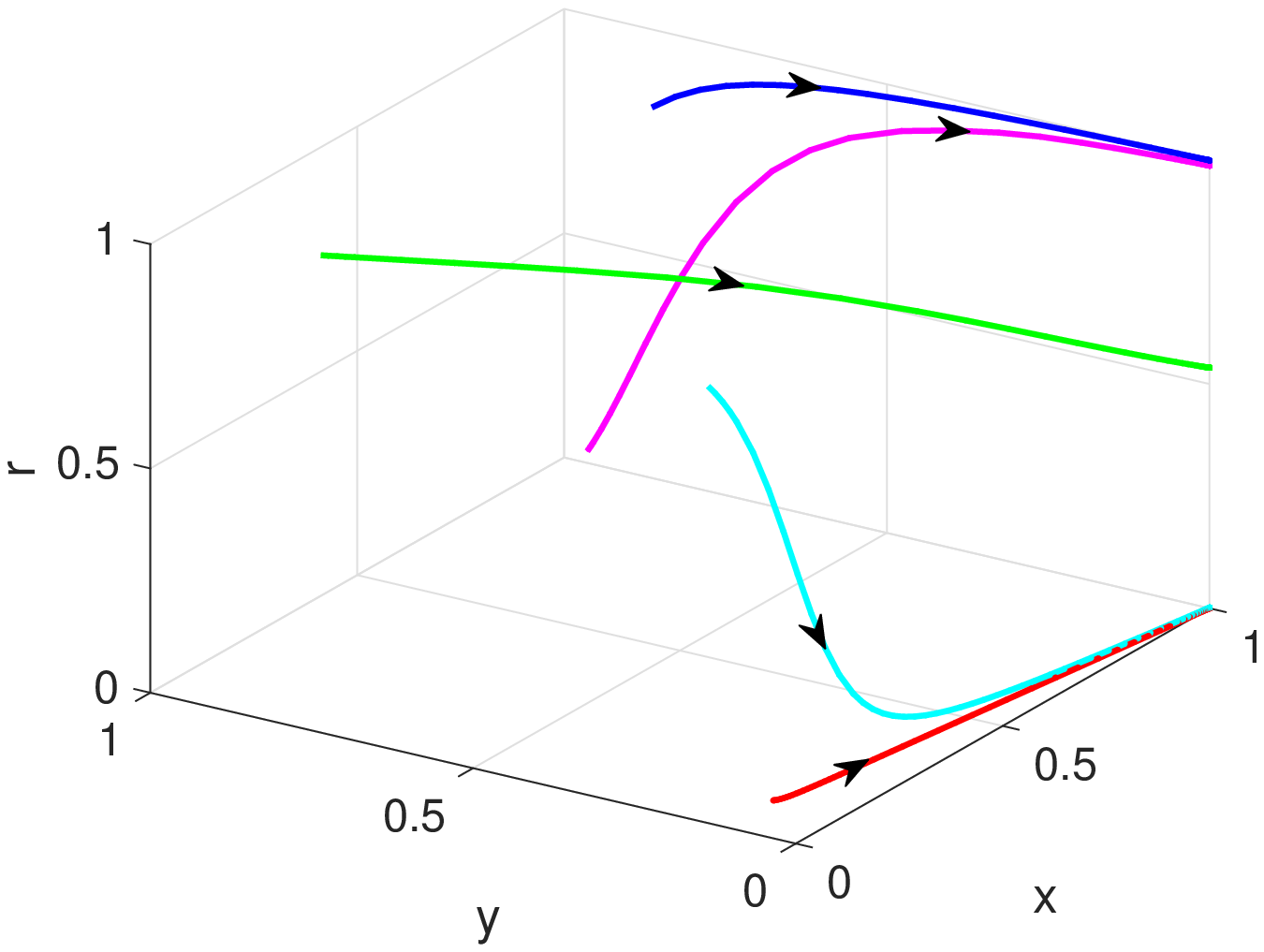}} \\
	\subfloat[$ \theta_{1}=1.5 $]{\includegraphics[width=.45\columnwidth]{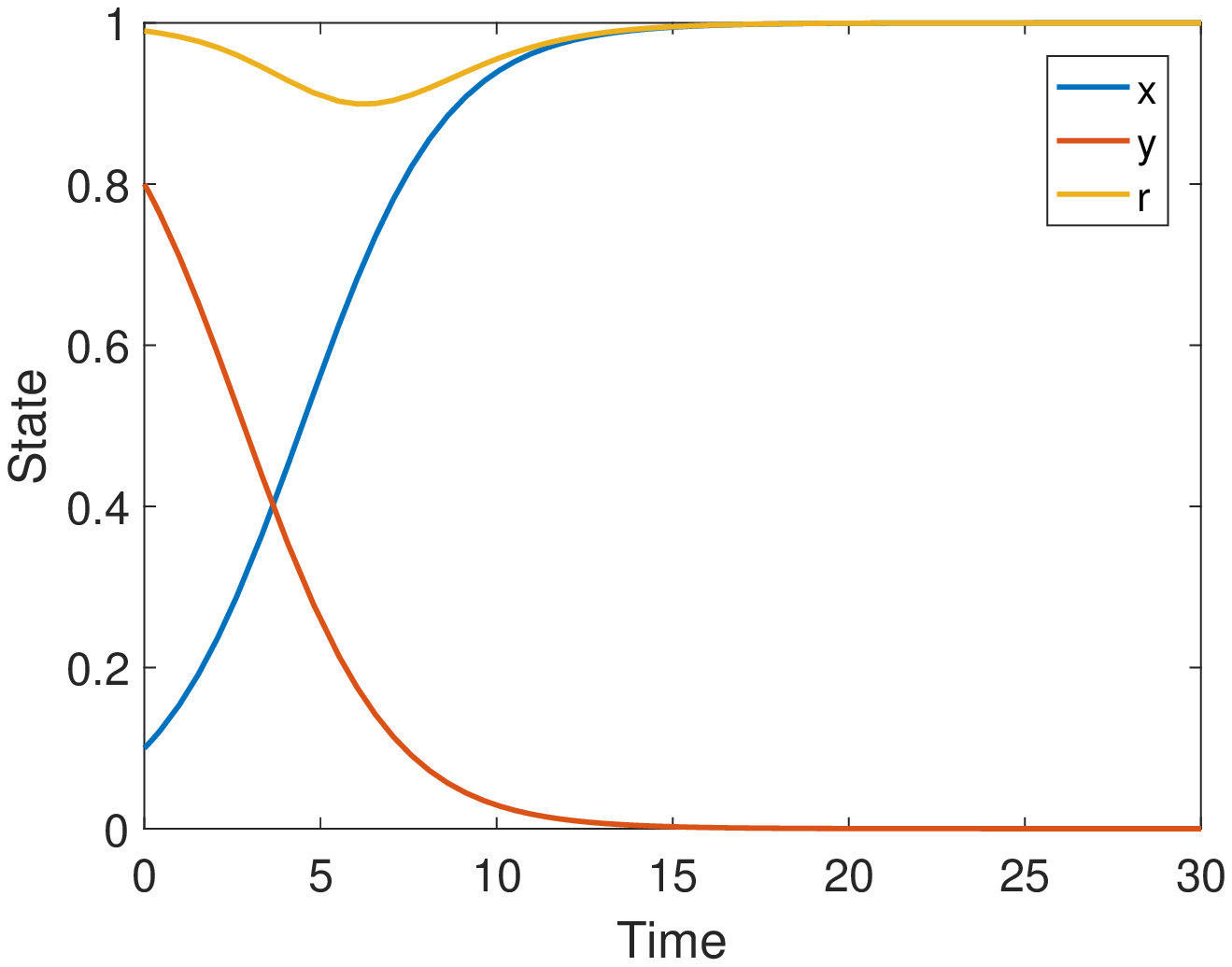} \quad\includegraphics[width=.45\columnwidth]{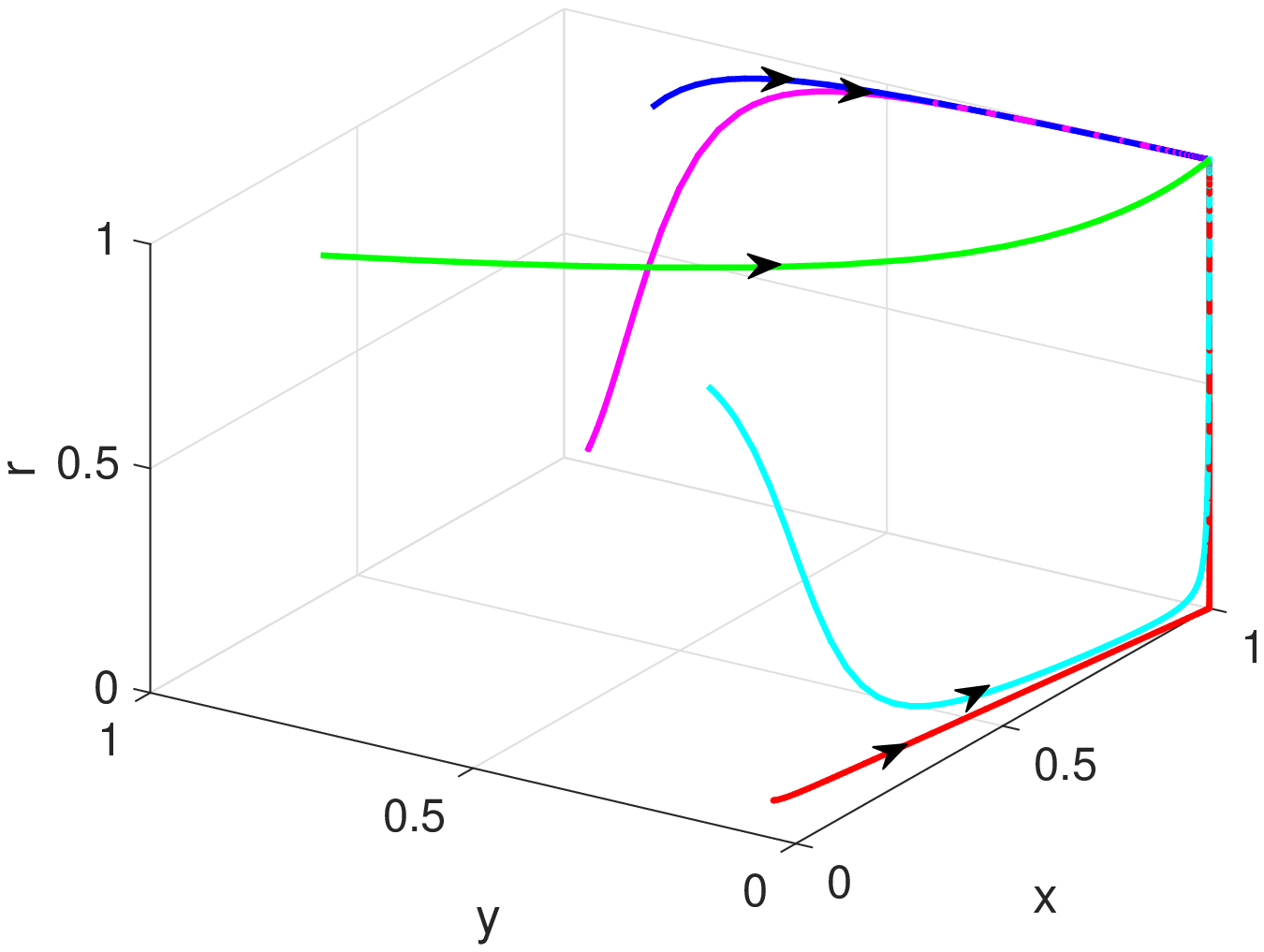}
	}\\
	\caption[Simulation 1]{Convergence of the trajectories when $\theta_{2}=0.5$ and $\theta_{1}=0.5$ (a), $\theta_{1}=1$ (b), $\theta_{1}=1.5 $ (c).  The left panels show the time evolution of  $(x,y,r)$ corresponding to the same initial condition $(0.1, 0.8, 0.99)$. In the right panels the distinct trajectories correspond to initial conditions $(0.1, 0.1, 0.01)$ (red), $(0.1, 0.8, 0.99)$ (green), $(0.1, 0.2, 0.9)$ (cyan), $(0.9, 0.9, 0.1)$ (magenta) and $(0.9, 0.9, 0.1)$ (blue).} 
	\label{fig:example1}
\end{figure}

Theorem \ref{theorem:convergence1} only gives a sufficient condition for convergence to the boundary point. 
It would be very difficult to construct the general conditions for non-convergence of the co-evolutionary system. In the following section, we will provide a specific type of payoff structure which results in oscillating behaviors. 

\subsection{Periodic Orbits}
We adopt the following payoff matrices which have been used in \cite{c19}. They are modified from the general Prisoner Dilemma (PD) game:
\begin{eqnarray}
A(r)=(1-r)\begin{bmatrix}
T_{1}&P_{1}\\
R_{1}&S_{1}
\end{bmatrix}+r\begin{bmatrix}
R_{1}&S_{1}\\
T_{1}&P_{1}
\end{bmatrix},\\
B(r)=(1-r)\begin{bmatrix}
T_{2}&P_{2}\\
R_{2}&S_{2}
\end{bmatrix}+r\begin{bmatrix}
R_{2}&S_{2}\\
T_{2}&P_{2}
\end{bmatrix},
\end{eqnarray}
where $ P_{1}>S_{1} $, $ T_{1}>R_{1}$ and  $ P_{2}>S_{2} $, $ T_{2}>R_{2} $ such that mutual cooperation is a Nash equilibrium when $ r \rightarrow
0 $ and mutual defection is a Nash equilibrium when $ r\rightarrow1 $. Intuitively, more players prefer to cooperate in the period of scarce resources and  incline to defect in the period of ample resources.
These environment-dependent payoff matrices can be rewritten as:
\begin{eqnarray*}
A(r)=\begin{bmatrix}
T_{1}-(T_{1}-R_{1})r&P_{1}-(P_{1}-S_{1})r\\
R_{1}+(T_{1}-R_{1})r&S_{1}+(P_{1}-S_{1})r
\end{bmatrix},\\
B(r)=\begin{bmatrix}
T_{2}-(T_{2}-R_{2})r&P_{2}-(P_{2}-S_{2})r\\
R_{2}+(T_{2}-R_{2})r&S_{2}+(P_{2}-S_{2})r
\end{bmatrix}.
\end{eqnarray*}

Using  (\ref{eqn:cosys1}), we arrive at 
\begin{equation}\label{eqn:cosys3}
\begin{cases}
&\dot{x}=x(1-x)[\delta_{PS_{1}}+(\delta_{TR_{1}}-\delta_{PS_{1}})y](1-2n)\\
&\dot{y}=y(1-y)[\delta_{PS_{2}}+(\delta_{TR_{2}}-\delta_{PS_{2}})x](1-2n)\\
&\dot{r}=r(1-r)[(1+\theta_{1})x+(1+\theta_{2})y-2],
\end{cases}
\end{equation}
where $ \delta_{PS_{1}}=P_{1}-S_{1}>0 $, $ \delta_{TR_{1}}=T_{1}-R_{1}>0 $ and $ \delta_{PS_{2}}=P_{2}-S_{2}>0 $, $ \delta_{TR_{2}}=T_{2}-R_{2}>0 $.

\subsubsection{Equilibria and stability}
Except for the eight isolated fixed points identified  in Lemma $1$ (i.e.,  eight corners of the domain), there is also a set of  fixed points in the interior of the cube, i.e., $ \left\lbrace (x,y,r):(1+\theta_{1})x+(1+\theta_{2})y=2, r=\dfrac{1}{2}\right\rbrace $. To study the local stability at these fixed points, one needs to calculate their Jacobians. It is easy to check that all the eight corner fixed points are unstable because each of their Jacobian matrices has at least one positive eigenvalue.

In addition, the Jacobian matrix of an interior fixed point $(x^{\ast},y^{\ast},r^{\ast})$ is of the following form:
\begin{equation}\label{eqn:jacobian}
J^{\ast}=\begin{bmatrix}
0&0&-2x^{\ast}(1-x^{\ast})[\delta_{PS_{1}}+(\delta_{TR_{1}}-\delta_{PS_{1}})y^{\ast}]\\
0&0&-2y^{\ast}(1-y^{\ast})[\delta_{PS_{2}}+(\delta_{TR_{2}}-\delta_{PS_{2}})x^{\ast}]\\
\dfrac{1+\theta_{1}}{4}&\dfrac{1+\theta_{2}}{4}&0
\end{bmatrix},
\end{equation} 
where $ (1+\theta_{1})x^{\ast}+(1+\theta_{2})y^{\ast}=2$.\\
In general the characteristic polynomial for a three dimensional  system takes the form
\begin{equation}\label{eqn:lambda}
\lambda^{3}-T\lambda^{2}-K\lambda-D=0,
\end{equation}
where $ T, D $ indicate the trace and determinant of the Jacobian respectively.
We calculate for the Jacobian (\ref{eqn:jacobian})
\[ T=\text{trace}(J^{\ast})=0, ~~~~D=\det(J^{\ast})=0.\]
and
\[
\begin{aligned}
K&=-\dfrac{1+\theta_{1}}{2}y^{\ast}(1-y^{\ast})[\delta_{PS_{2}}+(\delta_{TR_{2}}-\delta_{PS_{2}})x^{\ast}]\\
&-\dfrac{1+\theta_{2}}{2}x^{\ast}(1-x^{\ast})[\delta_{PS_{1}}+(\delta_{TR_{1}}-\delta_{PS_{1}})y^{\ast}]\\
&<0.
\end{aligned} \]
Thus, one can conclude the eigenvalues are $ \lambda_{1}=0 $ and $ \lambda_{2,3}=\pm \sqrt{-K}i $, for which $ \lambda_{2,3} $ are conjugate pure imaginary numbers.
When the real parts of all the eigenvalues are zero, one can only say that the interior equilibria  are neutrally stable for the linearized system. As a consequence, to analyze the stability of the original nonlinear system through  linearization  is not effective. Moreover, the  method in \cite{c19} utilizing the Hamiltonian system theory is not applicable anymore because  the system's dimension is odd.

Through simulation in MATLAB, we  find the trajectories initialized with $ (x_{0},y_{0},r_{0}) $ in the interior of the cube $ [0,1]^{3} $ and not at the fixed points will exhibit closed  periodic orbits as shown in Fig. 2. Therefore, it is natural to ask whether  all the solutions for system (\ref{eqn:cosys3}) in $(0,1)^{3}$ are either fixed points or periodic orbits.

\begin{figure}[htbp!]
	\centering 
	\subfloat[Time Evolution]{\includegraphics[width=.8\columnwidth]{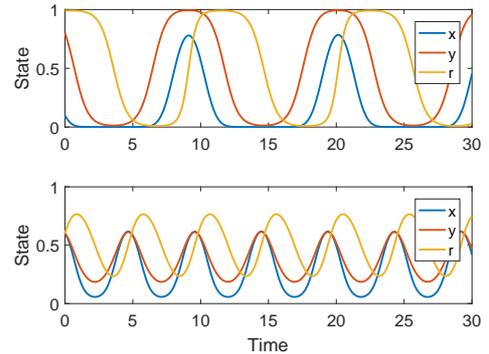}}\quad\subfloat[Periodic Orbits]{\includegraphics[width=.8\columnwidth]{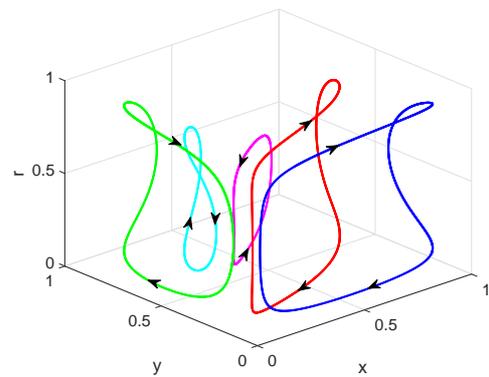}} 
	\caption[Simulation 2]{Simulation of Periodic Orbits with the parameters $R_{1}=2$, $S_{1}= -2$, $T_{1}=4$ and $P_{1}= 3$;  $R_{2}=3$, $S_{2}= 0$, $T_{2}=5$ and $P_{2}= 1$; $\theta_{1}=2$ and $\theta_{2}=2$. (a) Time evolution of states $(x,y,r)$ with initial conditions $(0.1, 0.8, 0.99)$ and $(0.6, 0.6, 0.6)$. (b) Periodic trajectories with initial conditions $(0.1, 0.1, 0.1)$ (red), $(0.1, 0.8, 0.9)$ (green), $(0.6, 0.6, 0.6)$ (magenta), $(0.2, 0.6, 0.6)$ (cyan) and $(0.9, 0.1, 0.1)$ (blue). 
    } 
	\label{fig:example2}
\end{figure}

In the following subsection, we are going to show that the neutral stability of linearized system at fixed points is preserved  in the original nonlinear system. We will first analyze the trajectories around the equilibrium points, and then apply an inverse derivation to prove this stability.

\subsubsection{Reversible system}

Before proving the claim on periodic orbits, we need to introduce the \emph{reversible system theory} \cite{c20,c21}. We say a transformation, denoted by $G$, is an \emph{involution} if the the composition of itself is the identify mapping, i.e. $ G\circ G = Identity$. 

\begin{definition}
	A dynamical system is said to be \emph{reversible}  if there is an involution in phase space which reverses the direction of time.
\end{definition}
Thus, a general system of coupled ordinary differential equations, 
\begin{equation} \label{eqn:reversible}
\dfrac{d\mathbf{x}}{dt}=F(\mathbf{x}),~~~~\mathbf{x}\in \mathbb{R}^{N},
\end{equation}
is reversible if there is an involution $ G $ which reverses the direction of time, i.e., 
\begin{equation}
\dfrac{dG(\mathbf{x})}{dt}=-F(G(\mathbf{x}))
\end{equation}
and hence 
\begin{equation}
dG|_{\mathbf{x}}\cdot F(\mathbf{x})=-F\circ G.
\end{equation}
The above definition implies that under the transformation of the $ N $-dimensional phase space by $ G $, the system (\ref{eqn:reversible}) is transformed to that obtained by just putting $ t\rightarrow -t $, so that under the combined action of involution and time reversal the equations are invariant. An involution that achieves this is called the \emph{time-reversal symmetry} of the system.

By the transformation of $ t\rightarrow -t $, $ x\rightarrow x $, $ y\rightarrow y $ and $ r\rightarrow 1-r $, one can easily verify the resulting system is the same as the original one.  So the  system  (\ref{eqn:cosys3}) is invariant under $ t\rightarrow -t $, with the phase space involution $ G:  x\rightarrow x ,  y\rightarrow y,   r\rightarrow 1-r $. Hence, this system is reversible with respect to the above $ G $ and the fixed manifold $ Fix(G) $ is the plane $ \{r=\frac{1}{2}\} $. Now we review  some intrinsic properties of the reversible system.

\begin{definition}[Symmetric orbits \cite{c22}] 
Let $ o(\mathbf{x}) $ be an orbit of a dynamical system, i.e. $ o(\mathbf{x})={ \varphi_{t}(\mathbf{x})|t\in \mathbb{R} } $. Then $ o(\mathbf{x}) $ is reversibly symmetric with respect to $ G $ when the orbit is set-wise invariant under $ G $, i.e. $ G(o(\mathbf{x})) =o(\mathbf{x})$.
\end{definition}
\begin{lem}[Periodic orbits for reversible systems \cite{c22}]
Let $ o(\mathbf{x}) $ be an orbit of the flow of an autonomous vector filed with time-reversal symmetry $ G $. Then,
\begin{enumerate}
\item  $ o(\mathbf{x}) $ is symmetric with respect to $ G $ if and only if $ o(\mathbf{x}) $ intersects $ \text{Fix}(G) $, in which case the orbit intersects $ \text{Fix}(G) $ in no more than two points and is fully contained in $ \text{Fix}(G^{2}) $.
\item  An orbit $ o(\mathbf{x}) $ intersects $ \text{Fix}(G) $ at precisely two points if and only if the orbit is periodic (and not a fixed point) and symmetric with respect to $ G $.
\end{enumerate}\label{theorem:periodicorbits}
\end{lem}
This lemma is widely used for studying reversible dynamical systems.

We continue to analyze  system (\ref{eqn:cosys3}) with the mentioned time-reversal symmetry $ G $. First, we show that there is one more invariant set for  system (\ref{eqn:cosys3}).

\begin{lem}[Invariant set] For  system (\ref{eqn:cosys3}), the open region, $ \Omega= (0,1)^{3} $, is an (positively and negatively) invariant set.
\end{lem}\label{lem:invariantset2}

\begin{proof}
First it can be easily checked in the plane $ \{r=1\}$ the point $ (0,0,1) $ is asymptotically stable for (\ref{eqn:cosys3}). 
Then, denote one trajectory starting from an arbitrary point $ p $ in the domain $\Omega= (0,1)^{3} $ by $ \varphi(t,0,p) $.  If it approaches the face $ \{r=1\}$ at some time $ t_{1} $, then it will always converges to the point $ (0,0,1) $ as $ t\rightarrow \infty $ because of the continuity of $\dot{x}$ and $\dot{y}$.  
Consider the deleted neighborhood $M$ of $ (0,0,1)$ in $\Omega$  , i.e., $M\subset \Omega$. One can check that $\dot{r}$ is always negative in $M$. Thus, the trajectory cannot approach the point $ (0,0,1) $.
One can also analogously prove that the trajectory $ \varphi(t,0,p) $  cannot reach other faces either. Hence, the positive invariance of $\Omega$ is proved.

If the trajectory starts from the region $ \mathbb{R}^{3}  \backslash \Omega $, it may reach the boundary of $[0,1]^{3}$.  But it cannot get into the interior, because the 6 faces of the cube are positively invariant. This corresponds to the negative invariance of  $  \Omega $. 

In conclusion, we have proved that the open set $ \Omega=(0,1)^{3} $ is positively and negatively invariant.
\end{proof}
Now we are ready to prove the main result.
\begin{theorem}
	The interior of the phase space, $ \Omega=(0,1)^{3} $, is filled with infinitely many independent periodic orbits, each centered at an  interior equilibrium point. 
\end{theorem}

\begin{proof}
First we divide the phase space $ \Omega=(0,1)^{3} $ into 4 regions by the two planes $ \{r=\frac{1}{2}\} $ and $ \{(1+\theta_{1})x+(1+\theta_{2})y=2\} $, which intersect at the line  of the interior equilibrium points.  Denote the four regions by
	\[ \begin{aligned}
	&\Omega_{1}: \{\dfrac{1}{2}<r<1, (1+\theta_{1})x+(1+\theta_{2})y>2 \};\\  &\Omega_{2}: \{\dfrac{1}{2}<r<1, (1+\theta_{1})x+(1+\theta_{2})y<2 \}; \\
	&\Omega_{3}: \{0<r<\dfrac{1}{2}, (1+\theta_{1})x+(1+\theta_{2})y<2\};\\  &\Omega_{4}: \{0<r<\dfrac{1}{2}, (1+\theta_{1})x+(1+\theta_{2})y>2 \}.\\ 
	\end{aligned} \]
From (\ref{eqn:cosys3}) it is easy to determine that the sign of the derivative $ \dot{r} $ is positive in $ \Omega_{1} $, $ \Omega_{4} $ and negative in $ \Omega_{2} $, $ \Omega_{3} $, respectively. Furthermore, the signs of the derivatives $ \dot{x},\dot{y} $ are negative  in $ \Omega_{1} $, $ \Omega_{2} $ and positive in $ \Omega_{3} $, $ \Omega_{4} $,\\
	Now, we consider an arbitrary trajectory $ \varphi(t,0,q) $ starting from a point $ q=(x_0,y_0,r_0) $ in  $ \Omega_{2} $.
	By using Lemma 5 one can exclude the trajectory going towards the face, as a result it will always stay inside the cube. \\
	Then, in $ \Omega_{2} $ we choose a  closed  subset $ S =\{\frac{1}{2}\leq r \leq r_0\} \cap \Omega_{2} $.  It is clear that the derivative  $ \dot{r} $ in  subset $ S $  remains negative. As the subset $ S $ is compact and the function on the right-hand side of equation $ \dot{r} $ is continuous, there must exist a small number $ a>0 $ such that $ \dot{r}\leq -a $. 
    
If the trajectory cannot go across the plane $ \{r=\frac{1}{2} \}$, then the projection of  $ \varphi(t,0,q) $ to the $ r $ axis at time $ \tau $ will be 
	\begin{equation*}
	r(\tau)=r(0)+\int_{0}^{\tau}\dot{z}dt \leq r(0)-a(\tau-0)
	\end{equation*}
	Then when the time $ \tau $ goes to  infinity, we get 
	\begin{equation*}
	\lim\limits_{\tau\rightarrow \infty}r(\tau) \leq r(0)-\lim\limits_{\tau\rightarrow \infty}a(\tau-0)= -\infty
	\end{equation*}
This result contradicts the boundedness of the subset. Therefore, the trajectory will traverse the plane $ \{r=\frac{1}{2} \}$ and always goes into the region $ \Omega_{3} $.\\
	Following  similar steps, one can verify the trajectory will cross the plane $ \{(1+\theta_{1})x+(1+\theta_{2})y=2\} $ by checking the signs of $ \dot{x} $ or $ \dot{y} $, and  enters into the region $ \Omega_{4} $, and then it continues to cross the plane $ \{r=\frac{1}{2}\} $ again. The geometric view of the whole process is depicted in Fig. (3).
	
In view of Lemma \ref{theorem:periodicorbits}, the trajectory will finally return to the starting point and form a closed orbit  because it intersects  $ Fix(G) $ at precisely two points. Hence, we have proved that each trajectory will form a  periodic orbit. Immediately one can claim that every closed orbit is neutrally stable since no other trajectories will converge to it.
\end{proof}

\begin{figure}[thpb]
	\centering
	\includegraphics[width=3in]{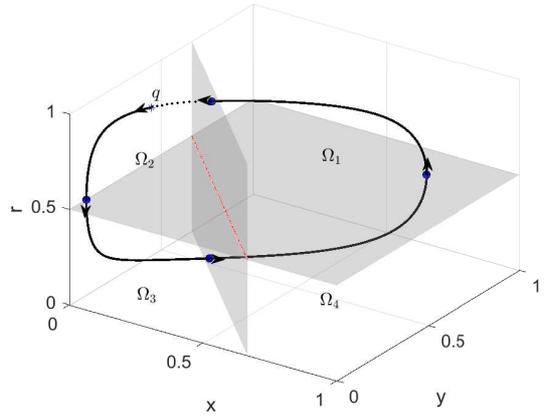}
	
	\caption{The process of forming a periodic orbit with initial point $q$.}
	\label{fig:periodo2}
\end{figure}

Then we can deduce the property of the internal equilibrium points.
\begin{corollary}
	The interior equilibrium points are neutrally stable for the nonlinear system (\ref{eqn:cosys3}).
\end{corollary}
By showing all the trajectories around the equilibria are independently periodic orbits, it is intuitive to get the neutrally stability of the corresponding equilibrium points.

The periodic trajectory of system (\ref{eqn:cosys3}) presents an unusual phenomenon, namely the state of two populations and environment oscillate dynamically. Under the payoff matrices of (10) and (11), the population profiles with the environment can evolve periodically without the extinction of any specify type strategic players.
The occurrence of closed orbits is more complicated than convergence to limit cycles or equilibria, because how the co-evolutionary dynamics evolve depends highly  on the system's initial states.

{\begin{remark}
When the payoff matrices are symmetric, namely $ A(r)=B(r) $, the dynamics of the resulting system will be mainly similar to the asymmetric case, but has another  feature, namely  there is a symmetric plane $ \{x=y \}$ and no trajectory can go across this plane. 

It can be easily verified that the new system  is invariant under the simple transformation $ x\rightarrow y $ and $ y\rightarrow x $. Thus the system is symmetric with respect to  the plane $ \{x=y\} $.
If there is a trajectory $\phi(t,0, x_0, y_0, r_0)$, there is always a symmetric trajectory $\phi(t, 0, y_0, x_0, r_0)$ for this system. If the trajectory  $\phi(t,0, x_0, y_0, r_0)$ crosses the plane $\{x=y \}$ at time $t_1$, it is easy to see that  at the same time $t_1$ the symmetric trajectory will also cross the plane $\{x=y\}$ at the same point. This result contradicts the fact that no two trajectories can cross  each other for an autonomous dynamical system. Hence, as a result no trajectory can go across the symmetric plane.
The intuitive interpretation of this property is that the proportion of cooperators in one population cannot exceed that of the other if there are fewer cooperators in this population in the beginning.
\end{remark}

%%%%%%%%%%%%%%%%%%%%%%%%%%%%%%%%%%%%%%%%%%%%%%%%%%%%%%%%%%%%%%%%%%%%%%%%%%%%%%%%
\section{CONCLUSIONS AND FUTURE WORKS}

We have proposed a new framework to study multi-population evolutionary game dynamics with environmental feedback and investigated whether and how convergence and coexistence take place in the new closed-loop system. The influence of dynamic and asymmetric payoff matrices have been studied in depth for two interacting populations. Two situations, convergence and dynamically coexistence, have been clarified through different approaches.
In particular, the scenarios where oscillation offers the best predictions of long-term behavior have been identified for the prisoner dilemma  game. The obtained results can be very useful to describe the evolution of multi-community societies in which individuals' payoffs and societal feedback interact.
In the future, we will generalize the framework to  multi-population situation. We are also interested in looking into networked populations. It is also of great interest to study different payoff matrices.

%%%%%%%%%%%%%%%%%%%%%%%%%%%%%%%%%%%%%%%%%%%%%%%%%%%%%%%%%%%%%%%%%%%%%%%%%%%%%%%%
%\section{ACKNOWLEDGMENTS}

%The authors gratefully acknowledge the contribution of National Research Organization and reviewers' comments.

%%%%%%%%%%%%%%%%%%%%%%%%%%%%%%%%%%%%%%%%%%%%%%%%%%%%%%%%%%%%%%%%%%%%%%%%%%%%%%%%

\end{document}